%% file: main.tex
\documentclass[conference,a4paper]{IEEEtran}

\IEEEoverridecommandlockouts

\usepackage{amsmath,amssymb,amsthm}
\usepackage{graphicx}
\usepackage{comment}
\usepackage{xcolor}
\usepackage[acronym]{glossaries}

\loadglsentries{glossary}

\newtheorem{proposition}{Proposition}

\begin{document}

\title{Physics-Compliant Modeling and Scaling Laws of\\Multi-RIS Aided Systems}

\author{\IEEEauthorblockN{
Matteo~Nerini\IEEEauthorrefmark{1},
Gabriele~Gradoni\IEEEauthorrefmark{2},
Bruno~Clerckx\IEEEauthorrefmark{1}}

\IEEEauthorblockA{\IEEEauthorrefmark{1}
Department of Electrical and Electronic Engineering, Imperial College London, London, U.K., m.nerini20@imperial.ac.uk}
\IEEEauthorblockA{\IEEEauthorrefmark{2}
Institute for Communication Systems, University of Surrey, Guildford, U.K.}

\thanks{This work was supported in part by UKRI under Grant EP/Y004086/1, EP/X040569/1, EP/Y037197/1, EP/X04047X/1, and EP/Y037243/1.}
}


\maketitle

\begin{abstract}
Reconfigurable intelligent surface (RIS) is a revolutionary technology enabling the control of wireless channels and improving coverage in wireless networks.
To further extend coverage, multi-RIS aided systems have been explored, where multiple RISs steer the signal toward the receiver via a multi-hop path.
However, deriving a physics-compliant channel model for multi-RIS aided systems is still an open problem.
In this study, we fill this gap by modeling multi-RIS aided systems through multiport network theory, and deriving the scaling law of the physics-compliant channel gain.
The derived physics-compliant channel model differs from the widely used model, where the structural scattering of the RISs is neglected.
Theoretical insights, validated by numerical results, show a significant discrepancy between the physics-compliant and the widely used models.
This discrepancy increases with the number of RISs and decreases with the number of RIS elements, reaching 200\% in a system with eight RISs with 128 elements each.
\end{abstract}

\glsresetall

\vskip0.5\baselineskip

\section{Introduction}
\label{sec:intro}

Reconfigurable intelligent surface (RIS) has emerged as a technology enabling dynamic control over the \gls{em} propagation environment in wireless networks \cite{wu21}.
RIS leverages surfaces made of elements with programmable reflecting properties to manipulate impinging \gls{em} signals.
While most of the literature on RIS focuses on systems aided by a single RIS, multi-RIS aided systems, also known as multi-hop RIS aided systems, have attracted attention as they can further enhance coverage.
Initial studies considered the optimization and analysis of two-RIS aided systems \cite{han20,zhe21-1}.
Besides, multi-RIS aided systems have been studied, where multiple cooperative RISs drive the \gls{em} signal to the intended receiver through a multi-hop path \cite{mei21,hua21,mei22-1,ma22,ngu23}.

Accurate modeling of wireless channels involving RIS is crucial for designing and optimizing RIS-aided systems.
To rigorously model wireless channels in the presence of a single RIS, multiport network analysis has been successfully utilized \cite{gra21,she20}.
Specifically, previous works derived physics-compliant RIS-aided channel models 
by using impedance parameters \cite{gra21} and scattering parameters \cite{she20}.
The relationship between impedance and scattering parameters has been more recently analyzed in \cite{nos24-2,li24,abr23,ner23}, and a model based on the admittance parameters has been derived in \cite{ner23}.

While substantial effort has been devoted to system-level optimization of multi-RIS aided systems, the rigorous modeling of multi-RIS aided channels is still an open issue.
The channel model widely used in previous works on multi-RIS systems still requires validation, and the assumptions under which it is applicable remain unexplored.
To solve this issue, in this study, we model the channel of multi-RIS aided systems through multiport network analysis, and derive the scaling law of the achievable channel gain.

The contributions of this study are outlined as follows.
\textit{First}, we derive a physics-compliant channel model for multi-RIS aided systems by using multiport network theory, clarifying its underlying assumptions.
Our model accounts for the structural scattering of the RISs, differently from the model widely used in the related literature.
\textit{Second}, we analyze the scaling law of the channel gain considering the derived model.
Specifically, we provide the expression of the average channel gain under \gls{los} channels in closed form.
\textit{Third}, we compare the derived physics-consistent model with the widely used model in terms of channel gain.
Theoretical derivations, supported by numerical simulations, show that the relative difference grows with the number of RISs and decreases with the number of RIS elements.
Remarkably, in a system aided by eight RISs, this difference can be as high as 2000\% with 16-element RISs, and as 200\% with 128-element RISs.

\section{Multiport Network Theory}
\label{sec:analysis}

Consider a communication system between a single-antenna transmitter and a single-antenna receiver aided by $L$ RISs, each having $N_I$ elements, as represented in Fig.~\ref{fig:diagram}.
Following previous literature \cite{gra21,she20,nos24-2,li24,abr23,ner23}, we model the wireless channel as an $N$-port network, with $N=2+LN_I$.

According to multiport network theory \cite[Chapter 4]{poz11}, the $N$-port network modeling the wireless channel can be characterized by its impedance matrix $\mathbf{Z}\in\mathbb{C}^{N\times N}$, which can be partitioned as
\begin{equation}
\mathbf{Z}=
\begin{bmatrix}
z_{TT} & \mathbf{z}_{TI} & z_{TR}\\
\mathbf{z}_{IT} & \mathbf{Z}_{II} & \mathbf{z}_{IR}\\
z_{RT} & \mathbf{z}_{RI} & z_{RR}
\end{bmatrix}.\label{eq:Z}
\end{equation}
$z_{TT}\in\mathbb{C}$ and $z_{RR}\in\mathbb{C}$ refer to the antenna self-impedance at the transmitter and receiver, respectively.
$\mathbf{Z}_{II}\in\mathbb{C}^{LN_I\times LN_I}$ can be partitioned as
\begin{equation}
\mathbf{Z}_{II}=
\begin{bmatrix}
\mathbf{Z}_{II,1} & \mathbf{Z}_{1,2} & \cdots & \mathbf{Z}_{1,L}\\
\mathbf{Z}_{2,1} & \mathbf{Z}_{II,2} & \cdots & \mathbf{Z}_{2,L}\\
\vdots & \vdots & \ddots & \vdots\\
\mathbf{Z}_{L,1} & \mathbf{Z}_{L,2} & \cdots & \mathbf{Z}_{II,L}
\end{bmatrix},
\end{equation}
where $\mathbf{Z}_{II,\ell}\in\mathbb{C}^{N_I\times N_I}$ refer to the impedance matrix of the antenna array at the $\ell$th RIS, whose diagonal entries refer to the antenna self-impedance while the off-diagonal entries refer to antenna mutual coupling, and $\mathbf{Z}_{i,j}\in\mathbb{C}^{N_I\times N_I}$ refer to the transmission impedance matrix from the $j$th RIS to the $i$th RIS.
Accordingly, $\mathbf{z}_{IT}\in\mathbb{C}^{LN_I\times 1}$ can be partitioned as
$\mathbf{z}_{IT}=[\mathbf{z}_{IT,1}^T,\mathbf{z}_{IT,2}^T,\ldots,\mathbf{z}_{IT,L}^T]^T$,
where $\mathbf{z}_{IT,\ell}\in\mathbb{C}^{N_I\times 1}$ is the transmission impedance matrix from the transmitter to the $\ell$th RIS.
$\mathbf{z}_{RI}\in\mathbb{C}^{1\times LN_I}$ can be partitioned as
$\mathbf{z}_{RI}=[\mathbf{z}_{RI,1},\mathbf{z}_{RI,2},\ldots,\mathbf{z}_{RI,L}]$,
where $\mathbf{z}_{RI,\ell}\in\mathbb{C}^{1\times N_I}$ is the transmission impedance matrix from the $\ell$th RIS to the receiver.
$z_{RT}\in\mathbb{C}$ refer to the transmission impedance matrix from the transmitter to the receiver.
Similarly, $\mathbf{z}_{TI}\in\mathbb{C}^{1\times LN_I}$, $\mathbf{z}_{IR}\in\mathbb{C}^{LN_I\times 1}$, and $z_{TR}\in\mathbb{C}$ refer to the transmission impedance matrices from the RISs to transmitter, and from the receiver to RISs, and from the receiver to transmitter, respectively.
For reciprocal wireless channels, it holds $z_{TR}=z_{RT}$, $\mathbf{z}_{TI}=\mathbf{z}_{IT}^T$, and $\mathbf{z}_{IR}=\mathbf{z}_{RI}^T$.

\begin{figure}[t]
\centering
\includegraphics[width=0.48\textwidth]{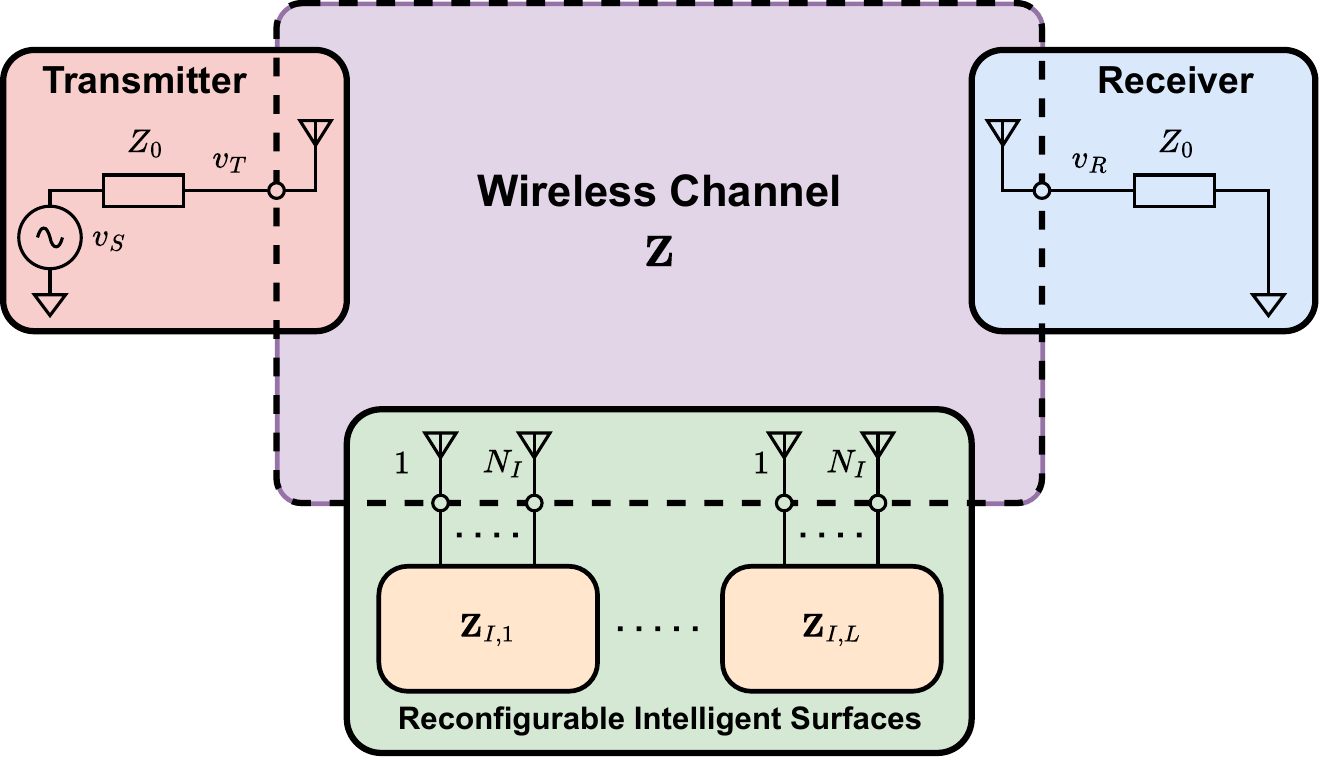}
\caption{Multi-RIS aided system modeled through multiport network theory.}
\label{fig:diagram}
\end{figure}

At the transmitter, the transmitting antenna is connected in series with a source voltage $v_{S}\in\mathbb{C}$ and a source impedance $Z_0$, e.g., set to $Z_0=50\:\Omega$, and we denote the voltage at the transmitting antenna as $v_{T}\in\mathbb{C}$.
At the $\ell$th RIS, the $N_I$ antennas are connected to an $N_I$-port reconfigurable impedance network with impedance matrix denoted as $\mathbf{Z}_{I,\ell}\in\mathbb{C}^{N_I\times N_I}$, for $\ell=1,\ldots,L$.
At the receiver, the receiving antenna is connected in series with a load impedance $Z_0$, and we denote the voltage at the receiving antenna as $v_{R}\in\mathbb{C}$.
In this study, our goal is to characterize the expression of the channel $h\in\mathbb{C}$ relating the transmitted signal $v_T$ and the received signal $v_R$ through
$v_R=hv_T$.

\section{Multi-RIS Aided Channel Model}
\label{sec:model}

As highlighted in previous work, the channel model for $h$ derived through multiport network theory with no assumptions is difficult to interpret \cite{gra21,she20,ner23}.
For this reason, we consider the following two assumptions, commonly accepted in the literature.
\begin{enumerate}
\item We assume sufficiently large transmission distances such that the currents at the transmitter/transmitter/RISs are independent of the currents at the RISs/receiver/receiver, respectively, also known as the unilateral approximation \cite{ivr10}, i.e., $\mathbf{z}_{TI}=\mathbf{0}$, $z_{TR}=0$, and $\mathbf{z}_{IR}=\mathbf{0}$.
\item We assume that the transmitting and receiving antennas are perfectly matched to the reference impedance $Z_0$, i.e., $z_{TT}=Z_0$ and $z_{RR}=Z_0$.
\end{enumerate}
Considering these two assumptions, it is possible to show that the channel $h$ is given by
\begin{equation}
h=\frac{1}{2Z_0}\left(z_{RT}-\mathbf{z}_{RI}\left(\mathbf{Z}_I+\mathbf{Z}_{II}\right)^{-1}\mathbf{z}_{IT}\right),\label{eq:H1}
\end{equation}
where $\mathbf{Z}_I\in\mathbb{C}^{LN_I\times LN_I}$ is a block diagonal matrix defined as
$\mathbf{Z}_{I}=\text{diag}(\mathbf{Z}_{I,1},\mathbf{Z}_{I,2},\ldots,\mathbf{Z}_{I,L})$,
having in the $\ell$th block the impedance matrix of the reconfigurable impedance network at the $\ell$th RIS \cite{gra21,ner23}.

Since the impact of the reconfigurable impedance matrices of the RISs $\mathbf{Z}_{I,\ell}$ is not apparent in \eqref{eq:H1} due to the matrix inversion operation, it is necessary to further simplify this model.
To this end, we consider with no loss of generality that the signal sent by the transmitter reaches the receiver by flowing from the $1$st RIS to the $L$th RIS.
Thus, we make the following four additional assumptions to obtain an interpretable channel model.
\begin{enumerate}
\setcounter{enumi}{2}
\item We assume that the channel between the transmitter and RIS $2,\ldots,L$ is blocked, i.e., $\mathbf{z}_{IT,\ell}=\mathbf{0}$, for $\ell=2,\ldots,L$, the channel between the receiver and RIS $1,\ldots,L-1$ is blocked, i.e., $\mathbf{z}_{RI,\ell}=\mathbf{0}$, for $\ell=1,\ldots,L-1$, and the channel between the transmitter and receiver is blocked, i.e., $z_{RT}=0$.
\item We assume large enough distances between the RISs such that the currents at the $i$th RIS are independent of the currents at the $j$th RIS, with $i<j$.
As in assumption~1), this is referred to as the unilateral approximation and allows us to set to zero the feedback channels between the RISs, i.e., $\mathbf{Z}_{i,j}=\mathbf{0}$, $\forall i<j$ \cite{ivr10}.
\item We assume that the $\ell$th RIS is only connected to the $(\ell-1)$th (if $\ell\neq1$) and $(\ell+1)$th (if $\ell\neq L$) RISs, for $\ell=1,\ldots,L$.
In other words, we assume that the channel between the $i$th and $j$th RIS is completely obstructed when $i-j\geq2$, i.e., $\mathbf{Z}_{i,j}=\mathbf{0}$ if $i-j\geq2$.
\item We assume perfect matching to $Z_0$ and no mutual coupling at all the $L$ RISs, i.e., $\mathbf{Z}_{II,\ell}=Z_0\mathbf{I}$, for $\ell=1,\ldots,L$, achievable by implementing the RISs via large reflectarrays with half-wavelength spacing.
\end{enumerate}
To simplify the channel model in \eqref{eq:H1} according to these assumptions, we introduce the matrix
\begin{equation}
\bar{\mathbf{Y}}=(\mathbf{Z}_I+\mathbf{Z}_{II})^{-1},
\end{equation}
partitioned as
\begin{equation}
\bar{\mathbf{Y}}=
\begin{bmatrix}
\bar{\mathbf{Y}}_{1,1} & \bar{\mathbf{Y}}_{1,2} & \cdots & \bar{\mathbf{Y}}_{1,L}\\
\bar{\mathbf{Y}}_{2,1} & \bar{\mathbf{Y}}_{2,2} & \cdots & \bar{\mathbf{Y}}_{2,L}\\
\vdots & \vdots & \ddots & \vdots\\
\bar{\mathbf{Y}}_{L,1} & \bar{\mathbf{Y}}_{L,2} & \cdots & \bar{\mathbf{Y}}_{L,L}
\end{bmatrix},
\end{equation}
where $\bar{\mathbf{Y}}_{i,j}\in\mathbb{C}^{N_I\times N_I}$, for $i,j=1,\ldots,L$.
Thus, the channel in \eqref{eq:H1} can be rewritten as
\begin{align}
h
&=\frac{1}{2Z_0}\left(z_{RT}-\sum_{i=1}^L\sum_{j=1}^L\mathbf{z}_{RI,i}\bar{\mathbf{Y}}_{i,j}\mathbf{z}_{IT,j}\right)\\
&=-\frac{1}{2Z_0}\left(\mathbf{z}_{RI,L}\bar{\mathbf{Y}}_{L,1}\mathbf{z}_{IT,1}\right),\label{eq:H2}
\end{align}
following Assumption~3), where the role of the transmission impedance matrices $\mathbf{z}_{RI,L}$ and $\mathbf{z}_{IT,1}$ is highlighted.
Furthermore, with assumptions~4), 5), and 6), the matrix $\mathbf{Z}_{II}$ simplifies as
\begin{equation}
\mathbf{Z}_{II}=
\begin{bmatrix}
Z_0\mathbf{I} &  &  &  & \mathbf{0}\\
\mathbf{Z}_{2,1} & Z_0\mathbf{I} &  &  & \\
 & \mathbf{Z}_{3,2} & \ddots &  & \\
 &  & \ddots & Z_0\mathbf{I} & \\
\mathbf{0} &  &  & \mathbf{Z}_{L,L-1} & Z_0\mathbf{I}\\
\end{bmatrix},
\end{equation}
having non-zero block matrices only in the diagonal and subdiagonal.
To accordingly simplify the channel model in \eqref{eq:H2}, we introduce the following proposition.

\begin{proposition}
Consider a square block matrix $\mathbf{M}\in\mathbb{C}^{LN\times LN}$ having square matrices $\mathbf{D}_\ell\in\mathbb{C}^{N\times N}$ in the diagonal, for $\ell=1,\ldots,L$, and square matrices $\mathbf{S}_{\ell,\ell-1}\in\mathbb{C}^{N\times N}$ in the subdiagonal, for $\ell=2,\ldots,L$, with all other blocks being zero matrices, i.e.,
\begin{equation}
\mathbf{M}=
\begin{bmatrix}
\mathbf{D}_1 &  &  &  & \mathbf{0}\\
\mathbf{S}_{2,1} & \mathbf{D}_{2} &  &  & \\
 & \mathbf{S}_{3,2} & \ddots &  & \\
 &  & \ddots & \mathbf{D}_{L-1} & \\
\mathbf{0} &  &  & \mathbf{S}_{L,L-1} & \mathbf{D}_L\\
\end{bmatrix}.\label{eq:M}
\end{equation}
If all $\mathbf{D}_\ell$ are invertible, the inverse of $\mathbf{M}$, denoted as $\mathbf{N}=\mathbf{M}^{-1}\in\mathbb{C}^{LN\times LN}$, is a square block matrix partitioned as
\begin{equation}
\mathbf{N}=
\begin{bmatrix}
\mathbf{N}_{1,1} & \mathbf{N}_{1,2} & \cdots & \mathbf{N}_{1,L}\\
\mathbf{N}_{2,1} & \mathbf{N}_{2,2} & \cdots & \mathbf{N}_{2,L}\\
\vdots & \vdots & \ddots & \vdots\\
\mathbf{N}_{L,1} & \mathbf{N}_{L,2} & \cdots & \mathbf{N}_{L,L}
\end{bmatrix},
\end{equation}
where $\mathbf{N}_{i,j}\in\mathbb{C}^{N\times N}$ is given by\footnote{Note that the index $k$ in the product in \eqref{eq:N} decreases from $i$ to $j+1$ since $i>j$.
The decreasing order matters because of the non-commutativity of matrix multiplication.}
\begin{equation}
\mathbf{N}_{i,j}=
\begin{cases}
\mathbf{0} & \text{if }i<j\\
\mathbf{D}_{i}^{-1} & \text{if }i=j\\
\left(-1\right)^{i-j}\mathbf{D}_{i}^{-1}
\prod_{k=i}^{j+1}\left(\mathbf{S}_{k,k-1}\mathbf{D}_{k-1}^{-1}\right) & \text{if }i>j\\
\end{cases}.\label{eq:N}
\end{equation}
\label{pro}
\end{proposition}
\begin{proof}
To prove the proposition, it is sufficient to notice that the matrix product $\mathbf{M}\mathbf{N}$ is the identity matrix.
\end{proof}

Since Proposition~\ref{pro} gives
\begin{multline}
\bar{\mathbf{Y}}_{L,1}=\left(-1\right)^{L-1}
\left(\mathbf{Z}_{I,L}+Z_0\mathbf{I}\right)^{-1}\\
\times\prod_{\ell=L}^{2}\left(\mathbf{Z}_{\ell,\ell-1}\left(\mathbf{Z}_{I,\ell-1}+Z_0\mathbf{I}\right)^{-1}\right),\label{eq:pro2}
\end{multline}
the channel in \eqref{eq:H2} can be further rewritten as
\begin{multline}
h
=-\frac{\left(-1\right)^{L-1}}{2Z_0}\Biggl(\mathbf{z}_{RI,L}\left(\mathbf{Z}_{I,L}+Z_0\mathbf{I}\right)^{-1}\Biggr.\\
\left.\times\prod_{\ell=L}^{2}\left(\mathbf{Z}_{\ell,\ell-1}\left(\mathbf{Z}_{I,\ell-1}+Z_0\mathbf{I}\right)^{-1}\right)\mathbf{z}_{IT,1}
\right),\label{eq:H3}
\end{multline}
explicitly emphasizing the impact of the transmission impedance matrices $\mathbf{Z}_{\ell,\ell-1}$ and the reconfigurable impedance matrices of the RISs $\mathbf{Z}_{I,\ell}$.
Remarkably, \eqref{eq:H3} gives the channel model of a multi-RIS aided system in the impedance parameters, or $Z$-parameters, of interest for less explored cascaded RIS-aided systems.

In the related literature, a RIS is often characterized through its scattering matrix \cite{she20}, which is related to the impedance matrix through a specific mapping according to microwave network theory \cite[Chapter 4]{poz11}.
Specifically, the scattering matrix of the $\ell$th RIS, denoted as $\boldsymbol{\Theta}_{\ell}\in\mathbb{C}^{N_I\times N_I}$, is related to $\mathbf{Z}_{I,\ell}$ through
\begin{equation}
\boldsymbol{\Theta}_{\ell}=\left(\mathbf{Z}_{I,\ell}+Z_0\mathbf{I}\right)^{-1}\left(\mathbf{Z}_{I,\ell}-Z_0\mathbf{I}\right),\label{eq:S}
\end{equation}
as discussed in \cite[Chapter 4]{poz11}.
Thus, by substituting
\begin{equation}
\left(\mathbf{Z}_{I,\ell}+Z_0\mathbf{I}\right)^{-1}=-\frac{1}{2Z_0}\left(\boldsymbol{\Theta}_{\ell}-\mathbf{I}\right),
\end{equation}
which is a direct consequence of \eqref{eq:S}, in \eqref{eq:H3}, and introducing the notation 
\begin{equation}
\mathbf{h}_{RI,L}=\frac{\mathbf{z}_{RI,L}}{2Z_0},\:
\mathbf{h}_{IT,1}=\frac{\mathbf{z}_{IT,1}}{2Z_0},\:
\mathbf{H}_{\ell,\ell-1}=\frac{\mathbf{Z}_{\ell,\ell-1}}{2Z_0},
\end{equation}
for $\ell=2,\ldots,L$, we obtain
\begin{equation}
h=
\mathbf{h}_{RI,L}\left(\boldsymbol{\Theta}_{L}-\mathbf{I}\right)
\prod_{\ell=L}^{2}\left(\mathbf{H}_{\ell,\ell-1}\left(\boldsymbol{\Theta}_{\ell-1}-\mathbf{I}\right)\right)\mathbf{h}_{IT,1},\label{eq:H}
\end{equation}
representing the channel model in the multi-RIS aided scenario illustrated in Fig.~\ref{fig:system}.
Interestingly, this scenario has been commonly studied in related literature on multi-RIS aided communications \cite{mei21,hua21,mei22-1,ma22,ngu23}.
However, the channel model in \eqref{eq:H} differs from the channel model widely used in related literature, which is instead given by
\begin{equation}
h^\prime=\mathbf{h}_{RI,L}\boldsymbol{\Theta}_{L}\prod_{\ell=L}^{2}\left(\mathbf{H}_{\ell,\ell-1}\boldsymbol{\Theta}_{\ell-1}\right)\mathbf{h}_{IT,1},\label{eq:Hprime}
\end{equation}
as employed in \cite{mei21,hua21,mei22-1,ma22,ngu23}.
Note that the only difference between \eqref{eq:H} and \eqref{eq:Hprime} lies in the fact that the terms $(\boldsymbol{\Theta}_{\ell}-\mathbf{I})$ in \eqref{eq:H} are replaced by the terms $\boldsymbol{\Theta}_{\ell}$ in \eqref{eq:Hprime}, for $\ell=1,\ldots,L$.
Remarkably, this is because the structural scattering effects of RIS are commonly neglected in related literature \cite{nos24-2}.

\begin{figure}[t]
\centering
\includegraphics[width=0.48\textwidth]{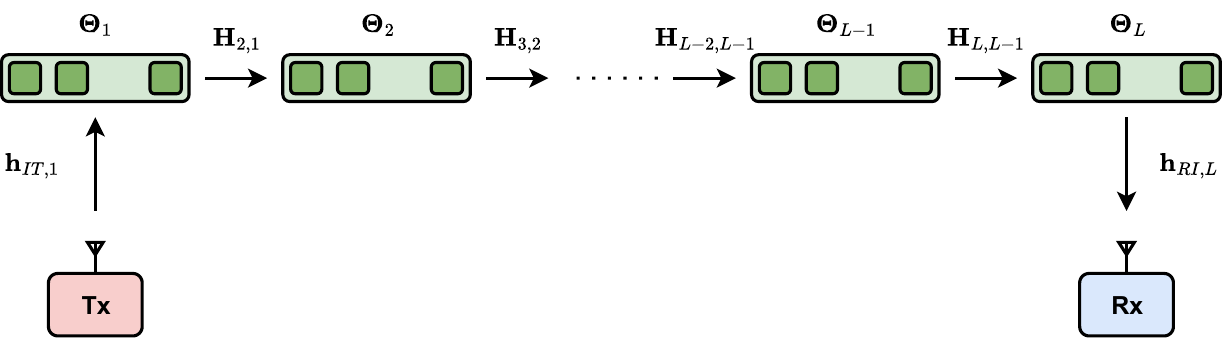}
\caption{Multi-RIS aided system.}
\label{fig:system}
\end{figure}

\section{Channel Gain Scaling Laws}
\label{sec:law}

To compare the physics-compliant channel model in \eqref{eq:H} and the widely used model in \eqref{eq:Hprime}, we derive the scaling laws of their channel gains, i.e., how they scale with $N_I$ as $N_I$ grows large.
To this end, we consider $\mathbf{h}_{RI,L}$, $\mathbf{h}_{IT,1}$, and $\mathbf{H}_{\ell,\ell-1}$, for $\ell=2,\ldots,L$, to be \gls{los} channels.
In this case, the channel between the $L$th RIS and the receiver writes as $\mathbf{h}_{RI,L}=\Lambda_{RI,L}\tilde{\mathbf{h}}_{RI,L}$, with $\Lambda_{RI,L}\in\mathbb{R}_+$ being the path gain and $\tilde{\mathbf{h}}_{RI,L}=[e^{j\varphi_{RI,L,1}},e^{j\varphi_{RI,L,2}},\ldots,e^{j\varphi_{RI,L,N_I}}]$.
Similarly, the channel between the transmitter and the $1$st RIS is given by $\mathbf{h}_{IT,1}=\Lambda_{IT,1}\tilde{\mathbf{h}}_{IT,1}$, with $\Lambda_{IT,1}\in\mathbb{R}_+$ being the path gain and $\tilde{\mathbf{h}}_{IT,1}=[e^{j\varphi_{IT,1,1}},e^{j\varphi_{IT,1,2}},\ldots,e^{j\varphi_{IT,1,N_I}}]^T$.
Besides, $\mathbf{H}_{\ell,\ell-1}=\Lambda_{\ell,\ell-1}\tilde{\mathbf{a}}_{\ell,\ell-1}\tilde{\mathbf{b}}_{\ell,\ell-1}$, with $\Lambda_{\ell,\ell-1}\in\mathbb{R}_+$ being the path gain, $\tilde{\mathbf{a}}_{\ell,\ell-1}=[e^{j\alpha_{\ell,\ell-1,1}},e^{j\alpha_{\ell,\ell-1,2}},\ldots,e^{j\alpha_{\ell,\ell-1,N_I}}]^T$, and $\tilde{\mathbf{b}}_{\ell,\ell-1}=[e^{j\beta_{\ell,\ell-1,1}},e^{j\beta_{\ell,\ell-1,2}},\ldots,e^{j\beta_{\ell,\ell-1,N_I}}]$, for $\ell=2,\ldots,L$.
Considering conventional RIS architectures, the scattering matrices of the RISs are in the form $\boldsymbol{\Theta}_{\ell}=\text{diag}(e^{j\theta_{\ell,1}},e^{j\theta_{\ell,2}},\ldots,e^{j\theta_{\ell,N_I}})$, with $\theta_{\ell,n_I}$ being the $n_I$th phase shift of the $\ell$th RIS, for $n_I=1,\ldots,N_I$ and $\ell=1,\ldots,L$.
In the following, we investigate the maximum channel gain achievable in the case of the physics-compliant and the widely used model.

\subsection{Physics-Compliant Channel Model}

Considering the channel model in \eqref{eq:H}, we can derive a global optimal solution for the scattering matrices $\boldsymbol{\Theta}_{\ell}$ to maximize the channel gain $\vert h\vert^2$ by adapting the optimization method proposed in \cite{mei21} for the widely used channel model.
Specifically, for a system with \gls{los} channels, \eqref{eq:H} can be expressed as
\begin{multline}
h=
\Lambda
\tilde{\mathbf{h}}_{RI,L}\left(\boldsymbol{\Theta}_{L}-\mathbf{I}\right)\tilde{\mathbf{a}}_{L,L-1}\\
\times\prod_{\ell=L-1}^{2}\left(\tilde{\mathbf{b}}_{\ell+1,\ell}\left(\boldsymbol{\Theta}_{\ell}-\mathbf{I}\right)\tilde{\mathbf{a}}_{\ell,\ell-1}\right)
\tilde{\mathbf{b}}_{2,1}\left(\boldsymbol{\Theta}_{1}-\mathbf{I}\right)\tilde{\mathbf{h}}_{IT,1},\label{eq:Hlos}
\end{multline}
where $\Lambda=\Lambda_{RI,L}\prod_{\ell=L}^{2}(\Lambda_{\ell,\ell-1})\Lambda_{IT,1}$ is the total path gain.
Interestingly, \eqref{eq:Hlos} can be rewritten as the product of scalar values $K_{\ell}\in\mathbb{C}$, i.e.,
\begin{equation}
h=
\Lambda\prod_{\ell=L}^{1}K_{\ell},
\end{equation}
where $K_{L}=\tilde{\mathbf{h}}_{RI,L}(\boldsymbol{\Theta}_{L}-\mathbf{I})\tilde{\mathbf{a}}_{L,L-1}$, $K_{1}=\tilde{\mathbf{b}}_{2,1}(\boldsymbol{\Theta}_{1}-\mathbf{I})\tilde{\mathbf{h}}_{IT,1}$, and $K_{\ell}=\tilde{\mathbf{b}}_{\ell+1,\ell}(\boldsymbol{\Theta}_{\ell}-\mathbf{I})\tilde{\mathbf{a}}_{\ell,\ell-1}$, for $\ell=2,\ldots,L-1$.
Thus, each $\boldsymbol{\Theta}_{\ell}$ can be globally optimized individually to maximize $\vert K_{\ell}\vert^2$, respectively, by setting
\begin{multline}
\theta_{L,n_I}=-\arg(\tilde{\mathbf{h}}_{RI,L}\tilde{\mathbf{a}}_{L,L-1})\\
-\arg([\tilde{\mathbf{h}}_{RI,L}]_{n_I})-\arg([\tilde{\mathbf{a}}_{L,L-1}]_{n_I}),
\end{multline}
\begin{multline}
\theta_{1,n_I}=-\arg(\tilde{\mathbf{b}}_{2,1}\tilde{\mathbf{h}}_{IT,1})\\
-\arg([\tilde{\mathbf{b}}_{2,1}]_{n_I})-\arg([\tilde{\mathbf{h}}_{IT,1}]_{n_I}),
\end{multline}
\begin{multline}
\theta_{\ell,n_I}=-\arg(\tilde{\mathbf{b}}_{\ell+1,\ell}\tilde{\mathbf{a}}_{\ell,\ell-1})\\
-\arg([\tilde{\mathbf{b}}_{\ell+1,\ell}]_{n_I})-\arg([\tilde{\mathbf{a}}_{\ell,\ell-1}]_{n_I}),
\end{multline}
for $\ell=2,\ldots,L-1$ and $n_I=1,\ldots,N_I$.

Consequently, the maximum channel gain achievable is given by
\begin{multline}
\left\vert h\right\vert^2=\Lambda^2\left(\left\vert\tilde{\mathbf{h}}_{RI,L}\tilde{\mathbf{a}}_{L,L-1}\right\vert+N_I\right)^2\\
\times\prod_{\ell=L-1}^{2}\left(\left\vert\tilde{\mathbf{b}}_{\ell+1,\ell}\tilde{\mathbf{a}}_{\ell,\ell-1}\right\vert+N_I\right)^2\left(\left\vert\tilde{\mathbf{b}}_{2,1}\tilde{\mathbf{h}}_{IT,1}\right\vert+N_I\right)^2,
\end{multline}
which depends on the specific channel realizations $\mathbf{h}_{RI,L}$, $\mathbf{h}_{IT,1}$, and $\mathbf{H}_{\ell,\ell-1}$, for $\ell=2,\ldots,L$.
To derive its expected value $\text{E}[\vert h\vert^2]$, we can regard the scalar products $\tilde{\mathbf{h}}_{RI,L}\tilde{\mathbf{a}}_{L,L-1}$, $\tilde{\mathbf{b}}_{\ell+1,\ell}\tilde{\mathbf{a}}_{\ell,\ell-1}$, and $\tilde{\mathbf{b}}_{2,1}\tilde{\mathbf{h}}_{IT,1}$ as independent sums of $N_I$ \gls{iid} complex random variables with unit modulus and uniformly distributed phase, i.e., with mean $0$ and variance $1$.
In this way, $\tilde{\mathbf{h}}_{RI,L}\tilde{\mathbf{a}}_{L,L-1}$, $\tilde{\mathbf{b}}_{\ell+1,\ell}\tilde{\mathbf{a}}_{\ell,\ell-1}$, and $\tilde{\mathbf{b}}_{2,1}\tilde{\mathbf{h}}_{IT,1}$ are independent and all distributed as $\mathcal{CN}(0,N_I)$ by the Central Limit Theorem, for $N_I$ large enough.
Thus, we have
\begin{equation}
\text{E}\left[\left\vert h\right\vert^2\right]=\Lambda^2\text{E}\left[\left(\left\vert c\right\vert+N_I\right)^2\right]^L,
\end{equation}
where $c\sim\mathcal{CN}(0,N_I)$.
By using $\text{E}[\vert c\vert]=\sqrt{\frac{\pi}{4}N_I}$ and $\text{E}[\vert c\vert^2]=N_I$, we obtain
\begin{equation}
\text{E}\left[\left\vert h\right\vert^2\right]=\Lambda^2\left(N_I^2+\sqrt{\pi N_I}N_I+N_I\right)^L,\label{eq:EG}
\end{equation}
giving the scaling law (for sufficiently large $N_I$) of the average channel gain for the physics-compliant model in closed form.

\subsection{Widely Used Channel Model}

Considering the widely used channel model in \eqref{eq:Hprime} in the case of a system with \gls{los} channels, we have
\begin{equation}
h^\prime=
\Lambda\prod_{\ell=L}^{1}K_{\ell}^\prime,
\end{equation}
where $K_{\ell}^\prime\in\mathbb{C}$ are given by $K_{L}^\prime=\tilde{\mathbf{h}}_{RI,L}\boldsymbol{\Theta}_{L}\tilde{\mathbf{a}}_{L,L-1}$, $K_{1}^\prime=\tilde{\mathbf{b}}_{2,1}\boldsymbol{\Theta}_{1}\tilde{\mathbf{h}}_{IT,1}$, and $K_{\ell}^\prime=\tilde{\mathbf{b}}_{\ell+1,\ell}\boldsymbol{\Theta}_{\ell}\tilde{\mathbf{a}}_{\ell,\ell-1}$, for $\ell=2,\ldots,L-1$.
As discussed in \cite{mei21}, the channel gain $\vert h^\prime\vert^2$ can be globally maximized by setting
\begin{align}
\theta_{L,n_I}&=-\arg([\tilde{\mathbf{h}}_{RI,L}]_{n_I})-\arg([\tilde{\mathbf{a}}_{L,L-1}]_{n_I}),\\
\theta_{1,n_I}&=-\arg([\tilde{\mathbf{b}}_{2,1}]_{n_I})-\arg([\tilde{\mathbf{h}}_{IT,1}]_{n_I}),\\
\theta_{\ell,n_I}&=-\arg([\tilde{\mathbf{b}}_{\ell+1,\ell}]_{n_I})-\arg([\tilde{\mathbf{a}}_{\ell,\ell-1}]_{n_I}),
\end{align}
for $\ell=2,\ldots,L-1$ and $n_I=1,\ldots,N_I$.
Accordingly, the maximum channel gain and its expected value are given by
\begin{equation}
\left\vert h^\prime\right\vert^2=\text{E}\left[\left\vert h^\prime\right\vert^2\right]=\Lambda^2N_I^{2L},\label{eq:Gprime}
\end{equation}
in agreement with previous literature \cite{mei21}.

\begin{figure}[t]
\centering
\includegraphics[width=0.32\textwidth]{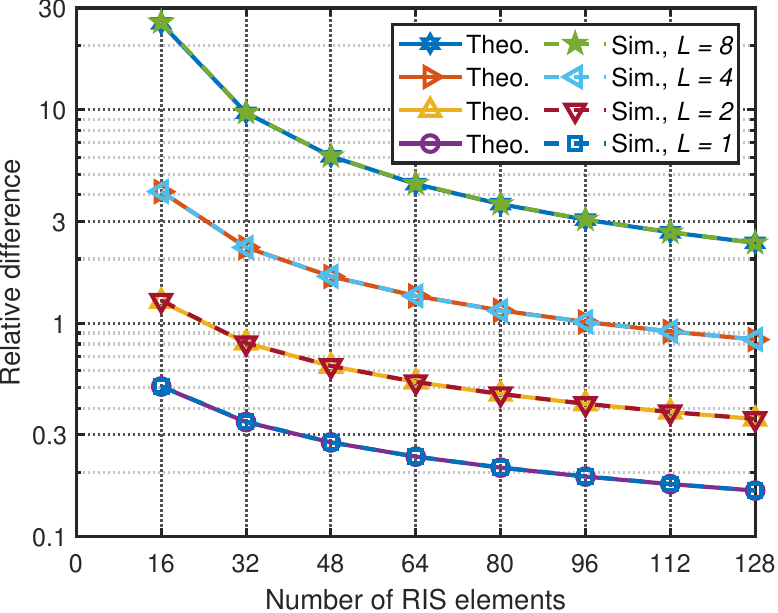}
\caption{Relative difference between the average channel gain with the physics-compliant model $\text{E}[\vert h\vert^2]$ and the widely used model $\text{E}[\vert h^\prime\vert^2]$.}
\label{fig:diff}
\end{figure}

\section{Numerical Results}
\label{sec:results}

We numerically quantify the difference between the channel models in \eqref{eq:H} and \eqref{eq:Hprime} by introducing the relative difference between the average channel gain for the widely used model $\text{E}[\vert h^\prime\vert^2]$ and physics-compliant model $\text{E}[\vert h\vert^2]$, defined as
\begin{equation}
\delta=\frac{\text{E}\left[\vert h\vert^2\right]-\text{E}\left[\vert h^\prime\vert^2\right]}{\text{E}\left[\vert h^\prime\vert^2\right]}.\label{eq:delta1}
\end{equation}
By substituting the scaling laws in \eqref{eq:EG} and \eqref{eq:Gprime} into \eqref{eq:delta1}, we finally obtain the relative difference $\delta$ as
\begin{equation}
\delta=\frac{\left(N_I+\sqrt{\pi N_I}+1\right)^L-N_I^{L}}{N_I^L}.\label{eq:delta2}
\end{equation}

In Fig.~\ref{fig:diff}, we report the relative difference $\delta$ derived theoretically in \eqref{eq:delta2} and obtained as a result of Monte Carlo simulations.
In our simulations, we consider $\varphi_{RI,L,n_I}$, $\varphi_{IT,1,n_I}$, $\alpha_{\ell,\ell-1,n_I}$, and $\beta_{\ell,\ell-1,n_I}$, independent and uniformly distributed in $[0,2\pi]$, for $\ell=2,\ldots,L$ and $n_I=1,\ldots,N_I$.
We observe that the theoretical insights are confirmed by the numerical results since the theoretical relative difference is the same as the simulated one.
In addition, we make the following two observations.
\textit{First}, the relative difference decreases with $N_I$ since in \eqref{eq:EG} the structural scattering term scales with $N_I$ while the RIS-aided term scales with $N_I^2$.
However, the relative difference is non-negligible for a practical number of RIS elements.
Specifically, considering $L=8$ RISs, the relative difference is higher than 2000\% when $N_I=16$ and higher than 200\% when $N_I=128$.
\textit{Second}, the relative difference increases with $L$ since each RIS contributes its structural scattering, which is not included in the widely used model.

\section{Conclusion}
\label{sec:conclusion}

We derive a physics-compliant channel model for multi-RIS aided systems, which differs from the channel model widely used in previous work since the RIS structural scattering is typically neglected in related literature.
We characterize the scaling law of the channel gain of our physics-compliant model and compare it with the widely used model.
Theoretical derivations, corroborated by numerical results, show that the physics-compliant channel gain substantially differs from the widely used channel gain, and that their relative difference can be as high as 2000\%.

\bibliographystyle{IEEEtran}
\bibliography{IEEEabrv,main}
\end{document}

%% file: glossary.tex
\newacronym{siso}{SISO}{single-input single-output}
\newacronym{star-ris}{STAR-RIS}{simultaneously transmitting and reflecting RIS}
\newacronym{iid}{i.i.d.}{independent and identically distributed}

\newacronym{los}{LoS}{line-of-sight}
\newacronym{mimo}{MIMO}{multiple-input multiple-output}

\newacronym{em}{EM}{electromagnetic}

%% file: main.bbl
\begin{thebibliography}{10}
\providecommand{\url}[1]{#1}
\csname url@samestyle\endcsname
\providecommand{\newblock}{\relax}
\providecommand{\bibinfo}[2]{#2}
\providecommand{\BIBentrySTDinterwordspacing}{\spaceskip=0pt\relax}
\providecommand{\BIBentryALTinterwordstretchfactor}{4}
\providecommand{\BIBentryALTinterwordspacing}{\spaceskip=\fontdimen2\font plus
\BIBentryALTinterwordstretchfactor\fontdimen3\font minus \fontdimen4\font\relax}
\providecommand{\BIBforeignlanguage}[2]{{%
\expandafter\ifx\csname l@#1\endcsname\relax
\typeout{** WARNING: IEEEtran.bst: No hyphenation pattern has been}%
\typeout{** loaded for the language `#1'. Using the pattern for}%
\typeout{** the default language instead.}%
\else
\language=\csname l@#1\endcsname
\fi
#2}}
\providecommand{\BIBdecl}{\relax}
\BIBdecl

\bibitem{wu21}
Q.~Wu, S.~Zhang, B.~Zheng, C.~You, and R.~Zhang, ``Intelligent reflecting surface-aided wireless communications: A tutorial,'' \emph{IEEE Trans. Commun.}, vol.~69, no.~5, pp. 3313--3351, 2021.

\bibitem{han20}
Y.~Han, S.~Zhang, L.~Duan, and R.~Zhang, ``Cooperative double-{IRS} aided communication: Beamforming design and power scaling,'' \emph{IEEE Wireless Commun. Lett.}, vol.~9, no.~8, pp. 1206--1210, 2020.

\bibitem{zhe21-1}
B.~Zheng, C.~You, and R.~Zhang, ``Double-{IRS} assisted multi-user {MIMO}: Cooperative passive beamforming design,'' \emph{IEEE Trans. Wireless Commun.}, vol.~20, no.~7, pp. 4513--4526, 2021.

\bibitem{mei21}
W.~Mei and R.~Zhang, ``Cooperative beam routing for multi-{IRS} aided communication,'' \emph{IEEE Wireless Commun. Lett.}, vol.~10, no.~2, pp. 426--430, 2021.

\bibitem{hua21}
C.~Huang~et al., ``Multi-hop {RIS}-empowered terahertz communications: A {DRL}-based hybrid beamforming design,'' \emph{IEEE J. Sel. Areas Commun.}, vol.~39, no.~6, pp. 1663--1677, 2021.

\bibitem{mei22-1}
W.~Mei and R.~Zhang, ``Intelligent reflecting surface for multi-path beam routing with active/passive beam splitting and combining,'' \emph{IEEE Commun. Lett.}, vol.~26, no.~5, pp. 1165--1169, 2022.

\bibitem{ma22}
X.~Ma, Y.~Fang, H.~Zhang, S.~Guo, and D.~Yuan, ``Cooperative beamforming design for multiple {RIS}-assisted communication systems,'' \emph{IEEE Trans. Wireless Commun.}, vol.~21, no.~12, pp. 10\,949--10\,963, 2022.

\bibitem{ngu23}
T.~V. Nguyen, D.~N. Nguyen, M.~D. Renzo, and R.~Zhang, ``Leveraging secondary reflections and mitigating interference in multi-{IRS/RIS} aided wireless networks,'' \emph{IEEE Trans. Wireless Commun.}, vol.~22, no.~1, pp. 502--517, 2023.

\bibitem{gra21}
G.~Gradoni and M.~Di~Renzo, ``End-to-end mutual coupling aware communication model for reconfigurable intelligent surfaces: An electromagnetic-compliant approach based on mutual impedances,'' \emph{IEEE Wireless Commun. Lett.}, vol.~10, no.~5, pp. 938--942, 2021.

\bibitem{she20}
S.~Shen, B.~Clerckx, and R.~Murch, ``Modeling and architecture design of reconfigurable intelligent surfaces using scattering parameter network analysis,'' \emph{IEEE Trans. Wireless Commun.}, vol.~21, no.~2, pp. 1229--1243, 2022.

\bibitem{nos24-2}
J.~A. Nossek, D.~Semmler, M.~Joham, and W.~Utschick, ``Physically consistent modelling of wireless links with reconfigurable intelligent surfaces using multiport network analysis,'' \emph{IEEE Wireless Commun. Lett.}, vol.~13, no.~8, pp. 2240--2244, 2024.

\bibitem{li24}
H.~Li, S.~Shen, M.~Nerini, M.~Di~Renzo, and B.~Clerckx, ``Beyond diagonal reconfigurable intelligent surfaces with mutual coupling: Modeling and optimization,'' \emph{IEEE Commun. Lett.}, vol.~28, no.~4, pp. 937--941, 2024.

\bibitem{abr23}
A.~Abrardo, A.~Toccafondi, and M.~Di~Renzo, ``Design of reconfigurable intelligent surfaces by using {S}-parameter multiport network theory – optimization and full-wave validation,'' \emph{IEEE Trans. Wireless Commun.}, 2024.

\bibitem{ner23}
M.~Nerini, S.~Shen, H.~Li, M.~Di~Renzo, and B.~Clerckx, ``A universal framework for multiport network analysis of reconfigurable intelligent surfaces,'' \emph{IEEE Trans. Wireless Commun.}, vol.~23, no.~10, pp. 14\,575--14\,590, 2024.

\bibitem{poz11}
D.~M. Pozar, \emph{Microwave engineering}.\hskip 1em plus 0.5em minus 0.4em\relax John wiley \& sons, 2011.

\bibitem{ivr10}
M.~T. Ivrlač and J.~A. Nossek, ``Toward a circuit theory of communication,'' \emph{IEEE Trans. Circuits Syst. I: Regul. Pap.}, vol.~57, no.~7, pp. 1663--1683, 2010.

\end{thebibliography}
